\documentclass[letterpaper, 10 pt, conference]{ieeeconf}  

\IEEEoverridecommandlockouts                              
\usepackage{float,amsmath,amsfonts,amssymb,amsthm,color,graphicx}
\usepackage{graphics} 
\usepackage{booktabs} 
\usepackage{caption, subcaption}
\usepackage{breqn}
\usepackage{array, mathtools}
\usepackage{pgfplots}
\usepackage{siunitx}
\usepackage{algorithm}
\usepackage[noend]{algpseudocode}
\usepackage{tabulary}
\usepackage{stmaryrd} 
\usepackage{cite}

\makeatletter
\newsavebox{\@brx}
\newcommand{\llangle}[1][]{\savebox{\@brx}{\(\m@th{#1\langle}\)}%
  \mathopen{\copy\@brx\kern-0.5\wd\@brx\usebox{\@brx}}}
\newcommand{\rrangle}[1][]{\savebox{\@brx}{\(\m@th{#1\rangle}\)}%
  \mathclose{\copy\@brx\kern-0.5\wd\@brx\usebox{\@brx}}}
\makeatother

\usepackage{tikz}
\usepackage{pgfplots}
\usetikzlibrary{decorations.markings}
\usetikzlibrary{patterns}
\usetikzlibrary{arrows}
\usetikzlibrary{shapes}
\usetikzlibrary{fit}
 \usetikzlibrary{calc}
\usetikzlibrary{decorations.pathreplacing}
\usetikzlibrary{arrows,positioning} 
\usetikzlibrary{pgfplots.groupplots}

\usepackage{thmtools}
\declaretheorem[style=definition,name=Definition,qed=$\blacksquare$]{definition}

\declaretheorem[style=definition,name=Assumption,qed=$\blacksquare$]{assumption}

\declaretheorem[style=definition,name=Remark,qed=$\blacksquare$]{remark}

\declaretheorem[style=definition,name=Problem Statement, numbered=no, qed=$\blacksquare$]{probstat}
\declaretheorem[style=definition,name=Lemma,qed=$\blacksquare$]{lemma}
\declaretheorem[style=plain,name=Proposition,qed=$\blacksquare$]{proposition}
\declaretheorem[style=plain,name=Theorem,qed=$\blacksquare$]{theorem}

\DeclareMathOperator*{\argmin}{arg\, min} 

\newcommand{\R}{\mathbb{R}}
\renewcommand{\S}{\mathcal{S}}
\newcommand{\K}{\mathcal{K}}

\newcommand{\X}{\mathcal{X}}

\newcommand{\W}{\mathcal{W}}

\newcommand{\rect}[1]{\llbracket #1 \rrbracket}
\newcommand{\corn}[1]{\llangle #1 \rrangle}
\newcommand{\cu}{\mathbf{u}}

\title{\LARGE {\bf Enforcing Safety at Runtime for Systems with Disturbances}}

\author{Matthew Abate and Samuel Coogan
\thanks{This work was partially supported by the Air Force Office of Scientific Research under Award No: FA9550-19-1-0015.}
\thanks{M. Abate is with the School of Mechanical Engineering and the School of Electrical and Computer Engineering, Georgia Institute of Technology, Atlanta, 30332, USA {\tt\small Matt.Abate@GaTech.edu}.}
\thanks{S. Coogan is with the School of Electrical and Computer Engineering and the School of Civil and Environmental Engineering, Georgia Institute of Technology, Atlanta, 30332, USA {\tt\small Sam.Coogan@GaTech.edu}.}
}

\begin{document}

\maketitle
\thispagestyle{empty}
\pagestyle{empty}

\begin{abstract}
Safety for control systems is often posed as an invariance constraint; the system is said to be safe if state trajectories avoid some unsafe region of the statespace for all time.  An assured controller is one that enforces safety online by filtering a desired control input at runtime, and control barrier functions (CBFs) provide an assured controller that renders a safe subset of the state-space forward invariant. Recent extensions propose CBF-based assured controllers that allow the system to leave a known safe set so long as a given backup control strategy eventually returns to the safe set, however, these methods have yet to be extended to consider systems subjected to unknown disturbance inputs.  

In this work, we present a problem formulation for CBF-based runtime assurance for systems with disturbances, and controllers which solve this problem must, in some way, incorporate the online computation of reachable sets.  In general, computing reachable sets in the presence of disturbances is computationally costly and cannot be directly incorporated in a CBF framework. To that end, we present a particular solution to the problem, whereby reachable sets are approximated via the mixed-monotonicity property.  Efficient algorithms exist for overapproximating reachable sets for mixed-monotone systems with hyperrectangles, and we show that such approximations are suitable for incorporating into a CBF-based runtime assurance framework. 
 
\end{abstract}

\section{Introduction}
Controllers whose safety guarantees are derived through the online enforcement of constraints, rather than \emph{a priori} verification, are referred to in literature as \emph{runtime assurance architectures} \cite{9029716} or \emph{active set invariance filters} (ASIF) \cite{thomas, thomas2}.  In this setting, system safety is posed as an invariance constraint, requiring that a system avoid some unsafe region of the statespace for all time. Specifications of this class are often used to describe real-word safety specifications due to the fact that the definition of real-world safety often is presented as the ability to avoid unsafe scenarios during deployment. 

Numerous mechanisms exist for enforcing invariance constraints, and in particular, control barrier functions (CBFs) are well suited for this task. CBF-based runtime assurance architectures modify a suggested desired input at runtime to create a safe forward invariant region in the state space.  This is a main idea of \cite{ames2014cdcCBFs, cbftutorial} where the resulting controller is formulated as a quadratic program for systems with no disturbances, and this idea is extended in \cite{ames_tac} to the setting with disturbances.  A limitation here is the need to verify a controlled forward invariant region \emph{a priori} and in general this region should be large; this problem can also be formulated as the search for a backup strategy with a corresponding controlled forward invariant region \cite{schouwenaars2006safe, bak2009system}. The authors of \cite{thomas, thomas2} present a CBF-based runtime assurance architecture, here formed via a verified backup strategy and safe region, which allows the system to leave the safe region.  The method eases the problem of verifying a forward invariant region \emph{a priori}, however, these works do not consider systems with disturbances.  In this work we present a problem formulation for CBF-based runtime assurance for controlled dynamical systems with disturbances, and we present an example solution to this problem where nondeterminism in the system model is assessed via the mixed-monotonicity property.

Mixed-monotone systems are separable via a decomposition function into increasing and decreasing components and this enables the approximation of reachable sets \cite{Coogan:2016, Invariance4MM, abate2020tight} and the identification of attractive and forward invariant sets \cite{Invariance4MM}; a similar approach is first pioneered in \cite{nonmonotone}, and we refer the reader also to \cite{smith2008monotone, monotonicity} for fundamental results on monotone dynamical systems. 

Efficient algorithms exist for overapproximating reachable sets for mixed-monotone systems with hyperrectangles, and we show that such approximations are suitable for incorporating into a CBF-based runtime assurance framework. As in \cite{thomas, thomas2}, our construction requires knowledge of a backup control strategy and a corresponding safe forward invariant region, however, the ASIF formed in this work allows the system to leave its safe region, and thus our construction does not require a large safe set \emph{a priori}.  
A main assumption in our approach is that closed-loop backup dynamics are mixed-monotone with respect to a known decomposition function; large classes of systems have been shown to be mixed-monotone with respect to closed-form decomposition functions constructed from, \emph{e.g.}, bounds on the system Jacobian matrix \cite{TIRA} or domains-specific knowledge \cite{7799445, smith2006discrete}, and in some instances decomposition functions can also be solved for by computing an optimization problem \cite{abate2020tight}.

In summary, the main contribution of this work are (a) we present a problem formulation for CBF-based runtime assurance for control systems with disturbances, and (b) we present a specific solution to the problem statement, whereby the nondeterminism in the system model is assessed through mixed-monotonicity based reachability methods.

This paper is structured as follows. We present our notation in Section II. In Section III we recall preliminary results on CBFs, and we also present a problem formulation for CBF-based runtime assurance for control systems with disturbances.  Throughout the remainder of the work, we present a solution to the problem statement, which relies on mixed-monotonicity based reachability methods.  To that end, we present preliminary results on mixed-monotone systems in Section IV, and we present an assured controller architecture in Section V which solves the problem statement and which accommodates nondeterminism in the system model via the mixed-monotonicity property.  We present a numerical example in Section VI, where we design and implement a runtime assurance architecture to enforce an interagent distance constraint on a platoon of vehicles.

\section{Notation}
We denote vector entries via subscript, \emph{i.e.}, $x_i$ for $i \in \{1,\cdots, n\}$ denotes the $i^{\text{th}}$ entry of $x\in \R^n $, and we denote the \emph{empty set} by $\varnothing := \{\}$.

Given $x,y\in \R^n$ with $x_i \leq y_i$ for all $i$,
\begin{equation*}
[x,\,y]:= \{z\in \R^n \,\mid\, x_i \leq z_i \leq y_i  \text{ for all } i\}
\end{equation*}
denotes the hyperrectangle with endpoints $x$ and $y$, and
\begin{equation*}
\corn{x,\, y} := \{z\in \R^n \,\mid\, z_i \in \{x_i,\, y_i\} \text{ for all } i\}
\end{equation*}
denotes the finite set of $2^n$ vertices of $[x,\, y]$.
We also allow $x_i\in\R \cup \{- \infty\}$ and $y_i\in \R \cup\{ \infty\}$ so that $[x,\,y]$ defines an \emph{extended hyperrectangle}, that is, a hyperrectangle with possibly infinite extent in some coordinates.  

Let $(x,\, y)$ denote the vector concatenation of $x,\, y \in \R^n$, \emph{i.e.}, $(x,\,y) := [x^T \, y^T]^T \in \R^{2n}$. 
Given $a=(x,y) \in \R^{2n}$ with $x_i \leq y_i$ for all $i$, we denote by $\rect{a}$ the hyperrectangle formed by the first and last $n$ components of $x$, \emph{i.e.},  $\rect{a}:=[x,\,y]$, and similarly $\corn{a} := \corn{x,\, y}$.

\section{Runtime Assurance for Nondeterministic Systems}
In this section, we define the problem of runtime assurance for continuous-time nondeterministic systems and provide a discussion on the problem statement.

\subsection{Problem Setting}
We consider controlled dynamical systems with disturbances of the form
\begin{equation}\label{AIDCDS}
    \dot{x} = f(x) + g_{1}(x)u + g_{2}(x)w
\end{equation}
with state $x \in \X \subseteq \R^n$, control input $u\in\mathbb{R}^m$, and Lipschitz continuous disturbance input $w \in \W\subset \mathbb{R}^n$. If $\W$ is a singleton set---equivalently, if the term $g_{2}(x)w$ is omitted from \eqref{AIDCDS}---then the system is said to be \emph{deterministic}; otherwise, the system is said to be \emph{nondeterministic}.

We let $\Phi(T;\, x,\, \cu,\, \mathbf{w})$ denote the state of \eqref{AIDCDS} at time $T \geq 0$, when starting from an initial state $x \in \X$ at time $0$ and evolving subject to a feedback controller $\cu : \X \rightarrow \R^m$ and the disturbance signal $\mathbf{w} : [0,\, T] \rightarrow \W$.

\begin{assumption}
We associate the system \eqref{AIDCDS} with an \textit{unsafe} subset of the system statespace $\X_{\rm u} \subset \X$.
\end{assumption}
A control policy is safe if it avoids the unsafe set as formalized next.

\begin{definition}\label{defsafe}
A controller $\cu : \X \rightarrow \R^m$ is \textit{safe with respect to state $x \in \X$} if $\Phi(T;\, x,\, \cu,\, \mathbf{w}) \in \X\setminus \X_{\rm u}$ for all $T \geq 0$ and for all $\mathbf{w} : [0,\, T] \rightarrow \W$.  We extend this notation to sets so that $\cu$ is \textit{safe with respect to $\S\subset \X$} if $\cu$ is safe with respect $x$ for all $x \in \S$.
\end{definition}

One way to establish safety is through invariance.

\begin{definition}\label{def1}
Given a controller $\cu$, a set $\S\subseteq \X$ is \emph{robustly forward invariant} for \eqref{AIDCDS} under $\cu$ if  $\Phi(T;  x,\, \cu,\, \mathbf{w})\in \S$ for all $x\in \S$, all $T\geq 0$ and all Lipschitz continuous disturbance inputs $\mathbf{w}:[0,\, T]\to \W$.
\end{definition}

\begin{remark}\label{rem1}
It is immediate that if $\S$ is robustly forward invariant for  \eqref{AIDCDS} under some control policy $\cu$ and $\S\cap \X_{\rm u}=\varnothing$ then $\cu$ is safe with respect to $\S$.
\end{remark}

Suppose $\S = \{x\in \X \,\vert\, h(x) \geq 0\} \subset \X\setminus\X_{\rm u}$ for some continuously differentiable $h : \R^n \rightarrow \R$ and consider the pointwise-defined controller 
\begin{align}
    &\qquad  \cu^{\rm CBF}(x) = \argmin_{u \in \R^m} ||u - \cu^{\rm d}(x)||_2^2
    \label{thing1}
    \\
    & \text{s.t. }\: \frac{\partial h}{\partial x}(x) (f(x) + g_1(x)u + g_2(x) w) \geq -\alpha(h(x)) \label{thing2}
    \\
    & \hspace{.75cm} \forall w \in \mathcal{W}  \nonumber
\end{align}
where $\alpha:\R\rightarrow \R$ is a given locally Lipschitz class-$\K$ function and $\cu^{\rm d}(x)$ is some  given controller. Provided the set of $u$ satisfying the constraint \eqref{thing2} is nonempty for all $x$, then $\mathcal{S}$ is robustly forward invariant for \eqref{AIDCDS} and $\cu^{\rm CBF}$ is safe with respect to $\S$ from Remark \ref{rem1}, and this statement is true even when $\cu^{\rm d}$ is not safe with respect to $\S$.  In this instance, $h$ is said to be a \emph{control barrier function (CBF)} for \eqref{AIDCDS} as developed in \cite{ames2014cdcCBFs}. The fundamental idea of the CBF formulation is that system safety is assured online by solving \eqref{thing1}--\eqref{thing2} to ensure $\mathcal{S}$ is robustly forward invariant. Note that, as formulated, for each $x$, \eqref{thing1}--\eqref{thing2} is a quadratic program with linear constraints, although there are potentially infinite constraints since \eqref{thing2} must hold for all $w\in \mathcal{W}$. However, in certain cases, it is possible to exchange \eqref{thing2} for a finite number of constraints. For example, if $\mathcal{W}$ is a polytope, as is the case below, then \eqref{thing2} need only be verified at the vertices of $\mathcal{W}$ since the constraint is affine in $w$. 

Applying $\cu^{\rm CBF}$ from \eqref{thing1}--\eqref{thing2} has added benefits beyond system safety and, in particular, $\cu^{\rm CBF}$ will evaluate to $\cu^{\rm d}$ whenever possible; thus, if $\cu^{\rm d}$ has performance advantages over $\cu$, then $\cu^{\rm CBF}$ will retain these advantages.

It is the primary focus of this paper to design safe controllers for the system \eqref{AIDCDS}.  To that end, we assume knowledge of a backup controller which is safe with respect to some subset of the statespace by virtue of a robustly invariant backup region as defined next.

\begin{definition}\label{def:backup}
The pair $(\cu^{\rm b},\, S_{\rm b})$ with $\cu^{\rm b}:\X \rightarrow \R^m$ and $S_{\rm b}=\{x \in \X \,\vert\, h(x)\geq 0\}\subset \X$ for a continuously differentiable $h:\X\to \mathbb{R}$ is a \emph{backup control policy} if: 
\begin{enumerate}
    \item $S_{\rm b}$ is compact and $S_{\rm b}\cap \X_{\rm u}=\varnothing$, 
    \item $h$ is concave on $\X$,
    \item $\frac{\partial h}{\partial x}\neq 0$  on the boundary of $S_{\rm b}$, and
    \item there exists a class-$\mathcal{K}$ function $\alpha:\R \rightarrow \R$ such that
\begin{align}
  \label{eq2}
  \frac{\partial h}{\partial x}(x)\left(f(x)+g_1(x)\cu^{\rm b}(x)+g_2(x)w\right)\geq -\alpha(h(x))
\end{align}
for all $x\in S_{\rm b}$ and for all $w\in\mathcal{W}$. 
\end{enumerate}
In particular, the last condition above implies $S_{\rm b}$ is robustly forward invariant for \eqref{AIDCDS}  under $\cu^{\rm b}$ via the CBF conditions discussed above and therefore $\cu^{\rm b}$ is safe with respect to $S_{\rm b}$ by virtue of the first condition \cite{ames_tac}. In this case, $\cu^{\rm b}$ is called a \emph{backup controller} and $S_{\rm b}$ its \emph{backup region}. \qedhere
\end{definition}

While applying the backup controller ensures system safety, there are two primary reasons why applying such a policy is generally not preferable:
\begin{enumerate}
    \item Backup controllers are typically designed without considering performance objectives.  In particular, another controller may exist which ensures safety and satisfies some performance objective. 
    \item $S_{\rm b}$ may not be well-developed, \emph{i.e.}, $\cu^{\rm b}$ may be safe with respect to a set larger than $S_{\rm b}$, and it is possible that $S_{\rm b}$ is too conservative to satisfy certain performance objectives. \qedhere
\end{enumerate}
We have already discussed how CBFs provide a solution to the first problem via, \emph{e.g.}, the controller \eqref{thing1}--\eqref{thing2}, in which knowledge of $\cu^{\rm b}$ is not even needed; see \cite{ames2014cdcCBFs, ames_tac} for further details. However, traditional CBF based controllers are still subject to the limitations of the second problem.
A solution to the second problem is presented in \cite{thomas, thomas2} for deterministic systems, where the authors effectively increase the size of the safe region through the use of look-ahead methods.

We now have the necessary tools to define the problem of runtime assurance for nondeterministic control systems.

\begin{probstat}[Runtime Assurance for Nondeterministic Control Systems]
Assume a system of form \eqref{AIDCDS} and a set of unsafe states $\X_{\rm u}\subset\X$.  Additionally, assume a backup control policy  $(\cu^{\rm b},\, S_{\rm b})$, and assume a desired controller $\cu^{\rm d}$ which satisfies some performance objective but is perhaps not safe with respect to $S_{\rm b}$.  The objective is to design a controller $\cu^{\rm ASIF}$ such that $\cu^{\rm ASIF}$ is safe with respect to $S_{\rm b}$ and such that $\cu^{\rm ASIF}(x)$ evaluates  to $\cu^{\rm d}(x)$ when it is safe to do so.
\end{probstat}

A controller $\cu^{\rm ASIF}$ which solves the problem statement is referred to as an \emph{assured controller} or an \emph{active set invariance filter} (ASIF).

\subsection{Discussion}
Note that the backup control policy itself is an assured controller when the performance control objectives are disregarded, \emph{i.e.} letting $\cu^{\rm ASIF}(x) = \cu^{\rm b}(x)$ for all $x \in \X$ we have that $\cu^{\rm ASIF}$ is safe with respect to $S_{\rm b}$. When performance control objectives are considered, one must incorporate the desired controller $\cu^{\rm d}$ in the ASIF formulation.
As such, the problem statement can be thought of as the task of integrating a backup strategy in an existing, perhaps unsafe, desired controller. 

We particularly aim for a solution that provides an assured controller that need not render $S_{\rm b}$ forward invariant; it may be the case, for instance, that for certain initial conditions $x \in S_{\rm b}$, the system \eqref{AIDCDS} will be driven out of $S_{\rm b}$ by $\cu^{\rm ASIF}$ and may not return.  Nonetheless, by virtue of the fact that $\cu^{\rm ASIF}$ is an assured controller we have that $\cu^{\rm ASIF}$ is safe with respect to $S_{\rm b}$ and, optimistically, it may be the case that $\cu^{\rm ASIF}$ is safe with respect to certain states outside of $S_{\rm b}$.

In Section \ref{mainresult}, we present a solution to the problem statement which allows the system to leave the, perhaps conservative, safe set $S_{\rm b}$.  In our proposed solution, we specifically address nondeterminism in the system model through mixed-monotonicity based reachability methods.

\section{Preliminaries on Mixed-Monotone Systems}\label{sec:MM}

Before visiting the general setting of \eqref{AIDCDS}, we first consider the nondeterministic autonomous system 
\begin{equation}\label{NAS}
    \dot{x} = F(x,\, w)
\end{equation}
and recall fundamental results in mixed-monotonicity theory.
As before, we let $\X$ and $\W$ denote the state and disturbance spaces of \eqref{NAS}, respectively, where we now assume $\X$ is an extended hyperrectangle and $\W$ is a hyperrectangle, with $\W := [\underline{w},\, \overline{w}]$ for $\underline{w},\, \overline{w}\in \R^m$ and $\underline{w}_i \leq \overline{w}_i$ for all $i$. 

\begin{definition}\label{def2}
Given a locally Lipschitz continuous function $d : \X \times \W \times \X \times \W \rightarrow \mathbb{R}^n$, the system \eqref{NAS} is \textit{mixed-monotone with respect to $d$}  if all of the following hold:
\begin{itemize}
    \item For all $x \in \X$ and all $w \in \W$, $d(x,\, w,\, x,\, w) = F(x,\, w)$.
    \item For all $i,\, j \in \{1,\, \cdots,\, n\}$ with $i \neq j$,  $\frac{\partial d_i}{\partial x_j}(x,\, w,\,  \widehat{x},\,  \widehat{w}) \geq 0$ for all $x,\,  \widehat{x} \in \X$ and all $w,\,  \widehat{w} \in \W$ whenever the derivative exists.
    \item For all $i,\, j \in \{1,\, \cdots,\, n\}$, $\frac{\partial d_i}{\partial  \widehat{x}_j}(x, w,  \widehat{x}, \widehat{w}) \leq 0$ for all $x,\,  \widehat{x} \in \X$ and all $w,\,  \widehat{w} \in \W$ whenever the derivative exists.
    \item For all $i\in \{1,\, \cdots,\, n\}$ and all $k \in \{1,\, \cdots,\, m\}$, 
    $\frac{\partial d_i}{\partial w_k}(x,\, w,\,  \widehat{x},\,  \widehat{w}) \geq 0$ and $\frac{\partial d_i}{\partial  \widehat{w}_k}(x,\, w,\,  \widehat{x},\,  \widehat{w}) \leq 0$
    for  all $x,\,  \widehat{x} \in \X$ and all $w,\,  \widehat{w} \in \W$ whenever the derivative exists. \qedhere
\end{itemize}
\end{definition}

If \eqref{NAS} is mixed-monotone with respect to $d$, $d$ is said to be a \emph{decomposition function} for \eqref{NAS}, and when $d$ is clear from context we simply say that \eqref{NAS} is mixed-monotone.  The mixed-monotonicity property is useful for, \emph{e.g.}, efficient reachable set computation, and these techniques have been applied in domains including transportation system \cite{7799445}, biological systems \cite{smith2006discrete}.  In these works, the authors construct decomposition functions from domain knowledge, however, it was recently shown in \cite{abate2020tight} that all systems of the form \eqref{NAS} are mixed-monotone and, thus, for all $F$ as in \eqref{NAS} there exists a $d$ satisfying the conditions of Definition \ref{def2}.  Nonetheless, identifying an appropriate decomposition function for ones particular setting still generally requires domain expertise, and we exemplify this point in the case study presented at the end of this work.

Let $\Phi^F(T; x,\mathbf{w})$ denote the state of \eqref{NAS} reached at time $T \geq 0$ starting from $x \in \X$ at time $0$ under the piecewise continuous input $\mathbf{w}:[0,\, T]\to \W$, and let
\begin{multline}
\label{reach1}
     R^{+}(T;\, \X_0) :=
     \Big{\{}\Phi^{F}(T;\, x,\, \mathbf{w}) \in \X \,\Big{|}\,  
     x\in \X_0
     \\ \text{for some } \mathbf{w} : [0,\, T] \rightarrow \W\Big{\}}
\end{multline}  
denote the time-$T$ forward reachable set of \eqref{NAS} from the set of initial conditions $\X_0 \subseteq \X$.
We next recall how over-approximations of reachable sets can be efficiently computed by considering a deterministic auxiliary system constructed from the decomposition function. 

Assume \eqref{NAS} is mixed-monotone with respect to $d$, and construct
\begin{equation}\label{eq:embedding}
\begin{bmatrix}
  \dot{x}\\
  \dot{ \widehat{x}} 
\end{bmatrix}
  = e(x,\, \widehat{x})
  := 
\begin{bmatrix}
  d (x,\, \underline{w},\,  \widehat{x},\,\overline{w})\\
  d ( \widehat{x},\,\overline{w},\, x,\, \underline{w}) 
\end{bmatrix}.
\end{equation}
The system \eqref{eq:embedding} is the \textit{embedding system} relative to $d$, and we let $\Phi^{e}( T;\, a)$ denote the state of this system at time $T \geq 0$ when initialized at $a \in \X\times\X$ at time $0$.  

\begin{proposition}[{\cite[Proposition 1]{Invariance4MM}} ]\label{prop:p1} 
Let $\X_0 = [\underline{x},\, \overline{x}]$ for some $\underline{x},\overline{x}$.
 If $\Phi^{e}( t;\, (\underline{x},\, \overline{x}))\in \X\times \X$ for all $0\leq t\leq T$, then $R^+(T;\, \X_0) \subseteq \rect{\Phi^{e}( T;\, (\underline{x},\, \overline{x}))}.$
\end{proposition}

By abuse of notation, we let $\Phi^e(T;\, x) := \Phi^e(T;\, (x,\,x))$, and thus it follows from Proposition \ref{prop:p1} that
\begin{equation}
    R^+(T;\, x) \subseteq \rect{\Phi^{e}( T;\, x)}.
\end{equation}
for all $x \in \X$ and all $T\geq 0$.

\section{Mixed-Monotonicity based Active Set Invariance}\label{mainresult}

In this section, we present a solution to the problem statement and design a controller architecture which both allows the system to leave $S_{\rm b}$ and ensures that the system never enters $\X_{\rm u}$.  The proposed controller uses a modified CBF formulation, where we now use mixed-monotonicity based reachability methods to assess the nondeterminism in the system model.

\subsection{Problem Formulation}
As prescribed in the problem statement, we assume a system of the form \eqref{AIDCDS}, an unsafe set $\X_{\rm u} \subset \X$, and a backup controller $\cu^{\rm b}$ with a compact backup region $S_{\rm b} = \{x \in \X \,\vert\, h(x) \geq 0\}$. We fix a desired controller $\cu^{\rm d}$ which is assumed to be preferable to the backup controller by some performance metric and, as in Section \ref{sec:MM}, we assume $\X$ is an extended hyperrectangle and $\W = [\underline{w},\, \overline{w}]$.

We denote by
\begin{equation}\label{backup_dyn}
    \dot{x} = F^{\rm b}(x,\, w) := f(x)+g_1(x) \cu^{\rm b}(x)+g_2(x)w
\end{equation}
the closed-loop dynamics of \eqref{AIDCDS} under $\cu^{\rm b}$ and we let $\Phi^{\rm b}(T;\, x,\, \mathbf{w}) := \Phi(T;\, x,\, \cu^{\rm b},\, \mathbf{w})$ denote the state transition function of this system.  
Thus, $h$ is a control barrier function for \eqref{backup_dyn} and $\cu^{\rm b}$ is safe with respect to $S^{\rm b}$.
Additionally, we denote by 
\begin{multline}
\label{eq:basin_of_attraction}
    S_{\rm b}^{+}(T) :=
     \Big{\{} x \in \X \,\Big{|}\,  
     \Phi^{\rm b}(T;\, x,\, \mathbf{w}) \in S_{\rm b}
     \\ \text{for all } \mathbf{w} : [0,\, T] \rightarrow \W\Big{\}}.
\end{multline}  
the time-$T$ basin of attraction of $S_{\rm b}$, which is the set of states in $\X$ that are guaranteed to enter $S_{\rm b}$ along trajectories of \eqref{backup_dyn} within the time horizon $[0,\, T]$. 

\begin{remark}
As a result of the fact that $S_{\rm b}$ is robustly forward invariant for \eqref{backup_dyn}, we additionally have that $S_{\rm b}^{+}(T)$ is robustly forward invariant for \eqref{backup_dyn} for all $T \geq 0$.
\end{remark}

As in \cite{thomas}, the ASIF formulation presented in this section allows the system to leave the safe set $S_{\rm b}$ in instances where the backup control policy is known to return the system to $S_{\rm b}$ on some finite time horizon.  For this reason, we associate the backup control policy $(\cu^{\rm b},\, S_{\rm b})$ with a fixed backup time $T_{\rm b}$, as formalised next.

\begin{assumption}\label{ass1}
The $T_{\rm b}$-second basin of attraction of $S_{\rm b}$ under the backup dynamics \eqref{backup_dyn} does not intersect the unsafe set, \emph{i.e.}, $S_{\rm b}^{+}(T_{\rm b}) \cap \X_{\rm u} = \varnothing.$ \qedhere
\end{assumption}

To verify Assumption \ref{ass1} holds, one can overapproximate backward reachable sets of $S_{\rm b}$ under \eqref{backup_dyn}, and check for intersection with the unsafe set $\X_{\rm u}$.  Many techniques allow for such an overapproximation and in the case study presented later, we implement one such method based on the mixed-monotonicity property.
Moreover, while we assume $T_{\rm b}$ is known \emph{a priori}, $S_{\rm b}^{+}(T_{\rm b})$ itself may be difficult to calculate in closed form. Thus, while a natural solution to the problem statement may be to construct a CBF-based ASIF to ensure the forward invariance of $S_{\rm b}^{+}(T_{\rm b})$, this solution may not be practically implementable when $S_{\rm b}^{+}(T_{\rm b})$ is not known.  The ASIF presented later in this section uses mixed-monotonicity based reachability methods to assesses whether or not the current system state is contained in $S_{\rm b}^{+}(T_{\rm b})$, and in this way we avoid an explicit description of $S_{\rm b}^{+}(T_{\rm b})$.

Lastly, we assume the backup dynamics \eqref{backup_dyn} are mixed-monotone.

\begin{assumption}\label{ass3}
The backup dynamics \eqref{backup_dyn} are mixed-monotone with respect to the decomposition function $d$, and we let $\Phi^e$ denote the transition function of its respective embedding system.
\end{assumption}

As discussed in the Introduction, Assumption \ref{ass3} is not especially restrictive since large classes of systems have been shown to be mixed-monotone with closed form expressions for the decomposition function $d$.

\subsection{Construction Methodology}
\label{sec:ideal}

Given $x\in \X$, possibly with $x \not\in S^{\rm b}$, our goal is to determine a suitable value $\cu^{\rm ASIF}(x)$; as suggested by the problem statement, $\cu^{\rm ASIF}(x)$ should be equal or close to $\cu^{\rm d}(x)$ if it is safe to do so. One method to determine whether or not $\cu^{\rm ASIF}(x)$ should be equal to $\cu^{\rm d}(x)$ is to assess the safety of the backup controller with respect to $x$, \emph{i.e.}, if $R_{\rm b}^+(T;\, x) \subseteq S^{\rm b}$ for some $T < T_{\rm b}$ then $\cu^{\rm ASIF}(x) = \cu^{\rm d}(x)$ is allowed, where we let $R_{\rm b}^+(T;\, x)$ denote the time-$T$ forward reachable set of \eqref{backup_dyn} as in \eqref{reach1}. 
We next present a family of functions that, for given $x\in \X$, can be used to assess whether or not $R_{\rm b}^+(T;\, x) \subseteq S^{\rm b}$ for some $T < T_{\rm b}$, and these functions exploit the mixed-monotonicity of \eqref{backup_dyn}.

Define
\begin{equation}\label{eq:6}
    \gamma^{\rm ideal}(T;\, x) :=  \inf_{z\in\rect{\Phi^e(T;\, x)}}h(z) =  \min_{z \in \corn{ \Phi^e(T;\, x)}}h(z),
\end{equation}
where the second equality comes from the concavity on $h$.
We show in the following lemma how $\gamma^{\rm ideal}$ is used to determine whether a state $x\in \X$ is contained in $S_{\rm b}^{+}(T)$ for given $T \geq 0$.

\begin{lemma}
For all $x \in \X$ and all $T \geq 0$,
\begin{equation}\label{cesar}
    \gamma^{\rm ideal}(T;\, x) \geq 0 \;\Rightarrow\; x \in S_{\rm b}^{+}(T).
\end{equation}
\end{lemma}
\begin{proof}
Fix $x \in \X$ and $T \geq 0$ such that $\gamma^{\rm ideal}(T;\, x) \geq 0$.  Then for all $z \in \rect{\Phi^e(T;\, x)}$ we have $h(z) \geq 0$, and thus $\rect{\Phi^e(T;\, x)} \subset S_{\rm b}$.  From Proposition \ref{prop:p1} we have $R_{\rm b}^{+}(T;\, x) \subseteq \rect{\Phi^e(T;\, x)}$, and therefore $\Phi^{\rm b}(T;\, x,\, \mathbf{w}) \in S_{\rm b}$  for all $\mathbf{w}$.  Therefore $x \in S_{\rm b}^{+}(T)$.
\end{proof}

Next define
\begin{equation}
  \label{eq:5}
\Psi^{\rm ideal}(x)=\sup_{0 \leq \tau\leq T_{\rm b}} \gamma^{\rm ideal}(\tau;\, x).  
\end{equation}
We show in the following proposition how $\Psi^{\rm ideal}$ is used to assess whether the backup control policy $\cu^{\rm b}$ is safe with respect to a given state.

\begin{proposition}\label{prop1}
If
\begin{equation}\label{eqpropone}
    \Psi^{\rm ideal}(x)
    \geq 0    
\end{equation}
for some $x \in \X$, then applying the backup control policy starting from $x$ at time $0$ ensures that there exists a time $T \leq T_{\rm b}$ such that $R_{\rm b}^+(t;\, x) \subseteq S_{\rm b}$ for all $t \geq T$.
In this case, we also have that $\cu^{\rm b}$ is safe with respect to $x$.
\end{proposition}
\begin{proof}
Assume that there exists an $x \in \X$ satisfying \eqref{eqpropone}. Then we have 
\begin{equation}\label{proop}
    \sup_{\tau\leq T_{\rm b}} \gamma^{\rm ideal}(\tau;\, x) \geq 0.
\end{equation}
Therefore, there must exist a time $T \leq T_{\rm b}$ such that $\gamma^{\rm ideal}(T;\, x) \geq 0$ and, at this time $R_{\rm b}^{+}(T;\, x) \subseteq S_{\rm b}$; see \eqref{cesar}.  Moreover, from Assumption \ref{ass1} the fact that $S_{\rm b}$ is robustly forward invariant on \eqref{backup_dyn}, we additionally have $R_{\rm b}^{+}(t;\, x) \subseteq S_{\rm b}$ for all $t \geq T$. 
\qedhere
\end{proof}

As a corollary to Proposition \ref{prop1}, note that the set
\begin{equation}\label{dutreix}
    S_{\Psi}^{\rm ideal} := \{x\in\X \,\vert\, \Psi^{\rm ideal}(x) \geq 0\}
\end{equation}
is robustly forward invariant on \eqref{backup_dyn}, and we have
\begin{equation}
    S_{\Psi}^{\rm ideal} \subseteq S_{\rm b}^{+}(T_{\rm b}).
\end{equation}

In summary, $\Psi^{\rm ideal}(x)$ is positive for states $x \in \X$ for which the backup controller is safe, and applying the backup controller to \eqref{AIDCDS} starting from $x$ ensures the system enters $S_{\rm b}$ on the time horizon $[0,\, T_{\rm b}]$. However, applying the backup controller may not be necessary; in fact, any control action that renders $S^{\rm ideal}_{\Psi}$ robustly forward invariant will be safe with respect to $x$. Control barrier functions are well suited for this task when the relevant functions are differentiable, however, $\gamma^{\rm ideal}$ and $\Psi^{\rm ideal}$ are generally not differentiable due to the $\min$ construction in \eqref{eq:6}.
In the next section, we present a novel soft-min construction of $\gamma^{\rm ideal}$ and $\Psi^{\rm ideal}$ which ensures differentiability.

\subsection{Barrier-Based ASIF Construction}\label{sec:MR}
We next present a differentiable relaxation of the functions $\gamma^{\rm ideal}$ and $\Psi^{\rm ideal}$, and these new functions are later incorporated in a control barrier function based ASIF.

We first recall the \textit{Log-Sum-Exponential} function.

\begin{definition}[Log-Sum-Exponential]\label{def:LSE}
We denote by
\begin{equation}
    \mbox{LSE}(\S,\, p) = -\frac{1}{p} \log \sum_{s \in \S} \mbox{exp}(-p\cdot s)
\end{equation}
the \textit{Log-Sum-Exponential} of the finite set $\S \subset \R$ with respect to the parameter $p > 0$.
\end{definition}

The Log-Sum-Exponential has several useful properties:
\begin{itemize}
    \item $\mbox{LSE}(\S,\, p)$ is differentiable with respect to the elements of $\S$, and
    \item $\mbox{LSE}(\S,\, p)$ approximates $\min \S$, \emph{i.e.},
    \begin{equation}\label{eq:minbound}
        \min \S -\frac{n}{p} \log 2\leq \mbox{LSE}(\S,\, p) < \min \S
    \end{equation}
    for all $p > 0$, and this approximation can be made {arbitrarily tight} by choosing $p$ large enough.
\end{itemize}

Next we introduce a continuously differently relaxation of $\Psi^{\rm ideal}$ and $\gamma^{\rm ideal}$ from the previous section.
To that end, fix $p > 0$ and consider 
\begin{equation}\label{gammafunct}
\begin{split}
    \gamma(t;\, x) 
    &:= \mbox{LSE}(\,\{h(z)\;\vert\; z \in \corn{\Phi^e(t;\, x)}\,\}\,,\, p)
    \\
    & \,= \frac{-1}{p}\log \sum_{z\in\corn{ \Phi^E(t;\, x)}} \exp(-p\cdot h(z)),
\end{split}
\end{equation}
where, from \eqref{eq:minbound} we have 
\begin{equation}
    \gamma^{\rm ideal}(t;\, x) - \frac{n}{p} \log 2 \leq \gamma(t;\, x)< \gamma^{\rm ideal}(t;\, x).
\end{equation}
Next define
\begin{equation}
    \label{eq:3}
    \Psi(x)=\sup_{0\leq\tau\leq T_{\rm b}} \gamma(\tau,x),
\end{equation}
and likewise $S_{\Psi} := \{x\in\X \,\vert\, \Psi(x) \geq 0\}$. Importantly, $\Psi(x)$ is differentiable with 
\begin{equation}
    \frac{\partial \Psi}{\partial x}(x) =  \frac{\partial \gamma}{\partial x}(\tau^*(x),x)
\end{equation}
where $\tau^*(x)$ is the maximizer from \eqref{eq:3}, \emph{i.e.}, $\tau^*(x)$ satisfies $\Psi(x)=\gamma(\tau^*(x),x)$, and this is a result of \cite[Theorem 1]{hogan1973directional}. 

In practice, $\frac{\partial \Psi}{\partial x}(x)$ is computed as follows. First, $\Phi^e(t,x)$ is computed for $t$ in the interval $[0,\, T_{\rm b}]$ by simulating the embedding dynamics \eqref{eq:embedding}, and the numerically simulated trajectory is used to identify the minimizer $\tau^*(x)$ for \eqref{eq:3}. Next, $\frac{\partial \Phi^e}{\partial x}(\tau^*(x),x)$ is computed numerically; for example, $n$ additional simulations of horizon $\tau^*(x)$ can be used to approximate the $n$ columns of the Jacobian matrix $\frac{\partial \Phi^e}{\partial x}$. Lasty, $\frac{\partial \gamma}{\partial x}(\tau^*(x),x)$ is obtained via the chain rule using prior computations.

\begin{lemma}\label{lem2}
$S_{\Psi}$ is a strict under-approximation of $S_{\Psi}^{\rm ideal}$, \emph{i.e.} $S_{\Psi} \subset S_{\Psi}^{\rm ideal}$.
\end{lemma}

The proof of Lemma \ref{lem2} is a direct result of \eqref{eq:minbound}.

As derived in Section \ref{sec:ideal}, $S_{\Psi}^{\rm ideal}$ is  robustly forward invariant on \eqref{backup_dyn}, however, $S_{\Psi}$ may not be. Further, $S_{\Psi}$ may not be robustly forward invariant under any control policy,  even though it is true that if  $\Psi(x)\geq 0$ for some $x$, then applying $\cu^{\rm b}$ will still result in eventually entering $S_{\rm b}$ within horizon $T_{\rm b}$. 
This is because it is no longer the case that applying $\cu^{\rm b}$ will keep $\Psi(x)$ from decreasing sometime before $x$ enters $S_{\rm b}$; $\Psi(x)$ could decrease by as much as $\frac{n}{p}\log 2$ due to the fact that $\gamma(t;\, x)$ is an under approximation of $\gamma^{\rm ideal}(t;\, x)$. Thus, even though a natural barrier-function-based reasoning might lead one to choose an input such that
\begin{align}
  \label{eq:4}
    \frac{d \Psi}{dt}(x(t))\geq -\alpha(\Psi(x(t)))
\end{align}
for some class-$\mathcal{K}$ function $\alpha:\R \rightarrow \R$ for all time, this may not be possible when $\Psi(x)$ is close to zero, and in particular, it may be the case that choosing $\cu^{\rm b}$ violates \eqref{eq:4}.
However, due to the fact that $S_\Psi \subset S_\Psi^{\rm ideal}$, if for some $x\in S_\Psi$ we have that $\cu^{\rm b}$ violates \eqref{eq:4}, then $\cu^{\rm b}$ is safe with respect to $x$ from Proposition \ref{prop1}, and thus it is acceptable to immediately switch to the backup control policy to retain safety.

\begin{algorithm}[t]
\caption{Runtime Assurance for Nondeterministic Control Systems}
\begin{algorithmic}[1]
\setlength\tabcolsep{0pt}
\Statex\begin{tabulary}{\linewidth}{@{}LLp{6cm}@{}}
\textbf{input}&:\:\:& Desired control policy $\cu^{\rm d}: \X\rightarrow\R^m$. \\
&:\:\:& Current State $x\in \X$.\\
&:\:\:& Class-$\K$ function $\alpha: \R\rightarrow\R$.\\
\textbf{output}&:\:\:& Assured control input $\cu^{\rm ASIF}(x) \in \R^m$.\\
&&
\end{tabulary}
\Function{$\cu^{\rm ASIF}(x)=$ASIF}{$\cu^{\rm d}$, $x$,\, $\alpha$}
\State \textbf{Compute:}  
\State $u^* = \argmin_{u \in \R^m} ||u - \cu^{\rm d}(x)||_2^2$
\State s.t. $\frac{\partial \Psi}{\partial x}(x) (f(x) + g_1(x)u + g_2(x)w)\geq -\alpha(\Psi(x))$
\State \hspace{0.45cm} $\forall w\in\corn{ \underline{w},\, \overline{w} }$
\If{Program feasible}
\State \textbf{return} $u^*$
\Else
\State \textbf{return} $\cu^{\rm b}(x)$
\EndIf
\EndFunction
\State\textbf{end function}
\end{algorithmic}
\label{alg:one}
\end{algorithm}

We next present our main result: an assured controller for nondeterministic control systems of the form \eqref{AIDCDS}.
This controller is presented in pseudocode (see Algorithm 1) and control actions are chosen point-wise in time.

Let $\Phi^{\rm ASIF}(T;\, x,\, \mathbf{w})$ denote the state of \eqref{AIDCDS} at time $T \geq 0$ when inputs are chosen using Algorithm \ref{alg:one} and when beginning from initial state $x \in \X$ at time $0$ and when subjected to the piecewise continuous inputs $\mathbf{w}$.

\begin{theorem}\label{thm1}
For all initial conditions $x \in S_{\rm b}$ and any Lipschitz continuous controller $\cu^{\rm d}:\X \rightarrow \R^m$, the controller $\cu^{\rm ASIF}$ from Algorithm \ref{alg:one} is such that $\Phi^{\rm ASIF}(T;\, x,\, \mathbf{w}) \not\in \X_{\rm u}$ for all $T \geq 0$. 
\end{theorem}
Theorem \ref{thm1} follows directly from Proposition \ref{prop1}, Lemma \ref{lem2}, and the preceding discussion; we thus omit a formal proof.

In summary, the assured controller $\cu^{\rm ASIF}$ defined by Algorithm \ref{alg:one}  (a) evaluates to the desired control input whenever possible, (b) allows the system \eqref{AIDCDS} to leave the safe region $S_{\rm b}$, and (c) ensures the system never enters the unsafe set $\X_{\rm u}$. Moreover, the optimization problem posed in Algorithm \ref{alg:one} contains only a finite number of affine constraints, where we note that the CBF constraint in Line 4 is only evaluated at the vertices of $\W$. Thus, the proposed assured controller can be computationally amenable to real-world applications, and we demonstrate the construction and implementation of such an ASIF through a case study provided in the next section.

\section{Numerical Example: Enforcing Inter-agent Distance Constraints on a Vehicle Platoon}
In this section we demonstrate the applicability of Algorithm 1 and design an ASIF which enforces inter-agent distance constraints on a platoon of vehicles.  

\subsection{Problem Setting}
Consider a platoon of $N \geq 2$ vehicles, whose velocity dynamics are given as
\begin{equation}
    \dot{x}_i = \beta x_i + a_i + w_i,
\end{equation}
where $x_i \in \R$ denotes the velocity of the $i^{\text{th}}$ vehicle, for $i \in \{1,\cdots,N\}$.  Here, $a_i$ denotes the acceleration of the vehicle, which is controlled by a global planner, $\beta \leq 0$ denotes a friction coefficient and $w \in \W \subset \R^N$ denotes a bounded additive noise term. 
We additionally let $p_i \in \R$ denote the position of the $i^{\text{th}}$ vehicle, so that $\dot{p}_i = x_i$.

Control decisions are made after referencing the relative displacements of vehicles in the platoon.  In particular, the accessible displacements are described by an undirected graph $\mathcal{G}$ with each node of the graph representing a vehicle and each edge of the graph denoting a displacement measurement between neighboring nodes.  We assume an arbitrary orientation of the edges of $\mathcal{G}$, so that the network is described by the \textit{incidence matrix} $A \in \R^{N \times K}$ with
\begin{equation*}
    A_{i, j} = \begin{cases}
    1 & \text{if vertex $i$ is the head of edge $j$} \\
    -1 & \text{if vertex $i$ is the tail of edge $j$} \\
    0 & \text{otherwise}
    \end{cases}
\end{equation*}
for $i\in \{1, \cdots, N\}$ and $j \in \{1, \cdots, K\}$ for a graph with $K$ edges. 
In this case, the vector containing the accessible displacements is given by $z = A p \in \R^{K}$, and the platoon dynamics then become
\begin{equation}\label{cssys}
    \begin{bmatrix}
      \dot{x} \\ \dot{z}
    \end{bmatrix}
    =
    \begin{bmatrix}
      \beta I & 0 \\ A^T & 0
    \end{bmatrix}    
    \begin{bmatrix}
      x \\ z
    \end{bmatrix}
    -
    \begin{bmatrix}
      D \\ 0
    \end{bmatrix}
    \cu(z)
    +\begin{bmatrix}
      w \\ 0
    \end{bmatrix}
\end{equation}
with control input $\cu(z) = [\,\cu_1(z_1),\, \cdots,\, \cu_K(z_K)\,]^T$.

While the theoretical results apply in the general setting of \eqref{cssys} with an arbitrary number of vehicles and links, for the remainder of the study, we consider a 3-cart instantiation of \eqref{cssys} with 2 control inputs, \emph{i.e.} we take $N = 3$ and $K=2$, and connectivity is given by
\begin{equation}\label{antonio}
    A = \begin{bmatrix}
    -1 & 0 \\
    1 & -1 \\
    0 & 1
    \end{bmatrix}.
\end{equation}
In this case $\X = \R^5$, and we fix $\beta = -1$ and $\W = [-0.1,\, 0.1]^3$.  This problem setting is shown in Figure \ref{fig3}.

\begin{figure}[t]
    \centering
    \includegraphics[width = .48\textwidth]{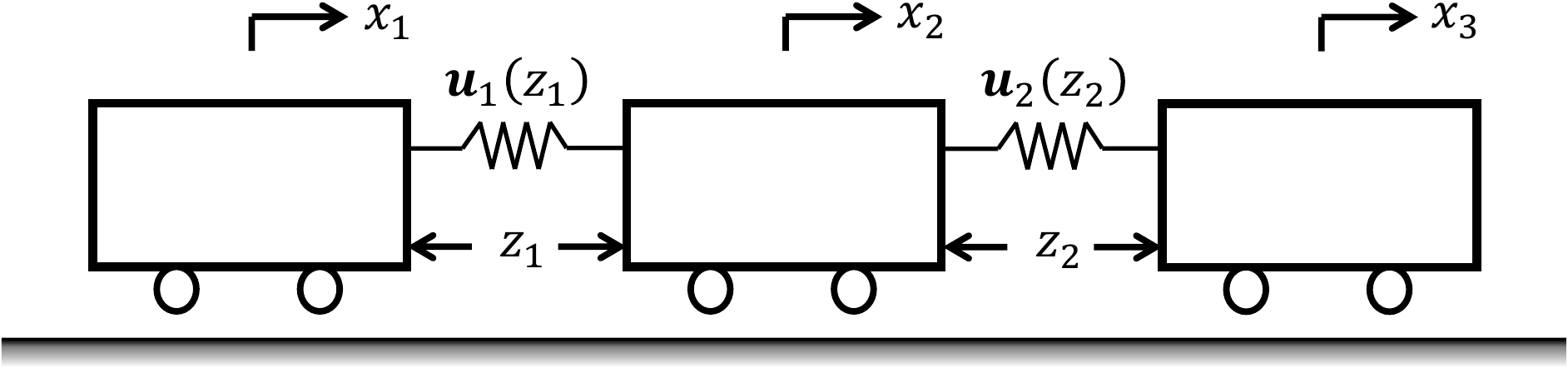}
    \caption{
    Problem setting.  $x_1,\, x_2,\, x_3 \in \R$ denote the vehicle velocities, and $z_1,\, z_2 \in \R$ denote the inter-agent distances when connectivity is given by \eqref{antonio}.  The control inputs $\cu_1,\, \cu_2$ effectively pull (push) the vehicles toward (away from) one another.
    }
    \label{fig3}
\end{figure}

We aim to enforce inter-agent distance constraints on \eqref{cssys} by applying the ASIF controller presented in Algorithm 1.  Specifically, we take an unsafe set
\begin{equation}
    \X_{\rm u} = \bigg{\{}
    \begin{bmatrix}
    x \\ z
    \end{bmatrix}
    \in \R^5 \:\Big{\vert}\: |z_1| \geq 8 \textbf{ or } |z_2| \geq 8 \bigg{\}},
\end{equation}
and we ignore vehicle collisions so that $z_1$ and $z_2$ are allowed to change sign over trajectories of \eqref{cssys}.

In the next section, we form a backup controller for \eqref{AIDCDS} which is safe by virtue of a forward invariant safe region.

\subsection{Constructing the Backup Controller}
We choose a backup controller
\begin{equation}\label{csbackup}
    \cu^{\rm b}(z) = \kappa \begin{bmatrix}
    \mbox{tanh}(\sigma z_{1}) \\
    \mbox{tanh}(\sigma z_{2})
    \end{bmatrix},
\end{equation}
with $\kappa = 2$ and $\sigma = 1/2$.  Roughly speaking, \eqref{csbackup} acts as two identical nonlinear springs which \textit{pull} the carts together when applied to \eqref{cssys}; by this description, $\kappa$ describes the maximum force which the springs apply before saturation, and $\sigma$ describes the distance at which the springs saturate.  The closed-loop dynamics of \eqref{cssys} under the backup control policy are
\begin{equation}\label{cs_CL}
    \begin{bmatrix}
      \dot{x} \\ \dot{z}
    \end{bmatrix}
    = F^{\rm b}
    \left(
    \begin{bmatrix}
      x \\ z
    \end{bmatrix},\, w \right) 
    =
    \begin{bmatrix}
     \beta x + w \\ A^T x
    \end{bmatrix}
    -
    \begin{bmatrix}
      A \kappa \\ 0
    \end{bmatrix}
    \begin{bmatrix}
    \mbox{tanh}(\sigma z_1) \\
    \mbox{tanh}(\sigma z_2)
    \end{bmatrix},
\end{equation}
and \eqref{cs_CL} is mixed-monotone on $\X$ with decomposition function
\begin{multline}
    d\Big{(}
    \begin{bmatrix}
    x \\ z
    \end{bmatrix},\, 
    \begin{bmatrix}
    \widehat{x} \\ \widehat{z}
    \end{bmatrix},\, w,\, \widehat{w}
    \Big{)}
    = 
    \begin{bmatrix}
      \beta x + w\\ (A^{+})^T x + (A^{-})^T \widehat{x}
    \end{bmatrix}
    \\
    -
    \begin{bmatrix}
      A^{-} \kappa \\ 0
    \end{bmatrix}
    \begin{bmatrix}
    \mbox{tanh}(\sigma z_1) \\
    \mbox{tanh}(\sigma z_2)
    \end{bmatrix}
    - 
    \begin{bmatrix}
      A^{+} \kappa \\ 0
    \end{bmatrix}
    \begin{bmatrix}
    \mbox{tanh}(\sigma \widehat{z}_1) \\
    \mbox{tanh}(\sigma \widehat{z}_2)
    \end{bmatrix}
\end{multline}
where $A^{+}$ and $A^{-}$ denote the positive and negative parts of $A$, respectively, and are given by
\begin{equation*}
A_{i, j}^{+} = 
    \begin{cases} 
    A_{i, j} & \text{if } A_{i, j} \geq 0 \\
    0 & \text{otherwise},
    \end{cases} \qquad
A_{i, j}^{-} = 
    \begin{cases} 
    A_{i, j} & \text{if } A_{i, j} < 0 \\
    0 & \text{otherwise}.
    \end{cases}
\end{equation*}

To construct a backup region $S_{\rm b}$, we consider a local linearization of \eqref{cs_CL}; that is, for \emph{small} disturbances, \eqref{cs_CL} locally behaves as
\begin{equation}\label{cs_lin}
    \begin{bmatrix}
      \dot{x} \\ \dot{z}
    \end{bmatrix}
    =
    \begin{bmatrix}
      \beta I & -A \kappa \sigma \\ A^T & 0
    \end{bmatrix}    
    \begin{bmatrix}
      x \\ z
    \end{bmatrix}.
\end{equation}
Further, \eqref{cs_lin} is asymptotically stable to the origin and is certified by the quadratic Lyapunov function
\begin{equation}\label{cs_lyap}
    V(x,\, z) = 
    \begin{bmatrix}
    x^T & z^T
    \end{bmatrix}
    P
    \begin{bmatrix}
    x \\ z
    \end{bmatrix}
\end{equation}
for
\begin{equation}
    P = 
    \begin{bmatrix}
    \kappa \sigma + AA^T &  -\beta A \\
    -\beta A^T & (\kappa^2 \sigma^2 + \beta^2) I + \kappa \sigma A^T A
    \end{bmatrix}.
\end{equation}
Thus, we consider an invariant safe set 
\begin{equation}\label{cs_safe_set}
    S_{\rm b} = 
    \bigg{\{}
    \begin{bmatrix} x \\ z\end{bmatrix} \in \R^5
    \,\bigg{\vert}\, 
    V(x,\, z) \leq \delta
    \bigg{\}}
\end{equation}
for appropriate $\delta \geq 0$.  
For the parameters taken in this study, $S_{\rm b}$ from \eqref{cs_safe_set} was verified to be robustly forward invariant on \eqref{cs_CL} when $\delta = 9/4$.

Let
\begin{multline}
\label{reach2}
    R_{\rm b}^{-}(T;\, \X_1) :=
    \Big{\{}
	x \in \X
	\,\Big{|}\,
     \Phi^{F}(T;\, x,\, \mathbf{w}) \in \X_1\\
        \text{ for some }\mathbf{w} : [0,\, T] \rightarrow \W\Big{\}}
\end{multline}  
denote the time-$T$ backward reachable set of \eqref{cs_CL}.
We next calculate a backup horizon $T_{\rm b}$ such that 
\begin{equation}\label{safeset_cond}
    S_{\rm b}^+(T_{\rm b}) \cap \X_{\rm u} = \varnothing,
\end{equation}
and this is done by showing that $R_{\rm b}^{-}(T;\, S_{\rm b}) \cap \X_{\rm u} = \varnothing$.
In particular, we overapproximate $R_{\rm b}^{-}(1;\, S_{\rm b})$ using \cite[Proposition 2]{Invariance4MM} and find that $R_{\rm b}^{-}(1;\, S_{\rm b}) \cap \X_{\rm u} = \varnothing$.
Therefore we take a backup time $T_{\rm b} = 1$ which satisfies \eqref{safeset_cond}.

We now have the necessary tools to implement Algorithm 1. We demonstrate the creation and application of the active set invariance filter in the next section.

\subsection{Simulated Implementation}
We next construct an ASIF to assure the system \eqref{cssys}, where we take the backup controller $\cu^{\rm b}$ from \eqref{csbackup}, safe set $S_{\rm b}$ from \eqref{cs_safe_set}, and backup time $T_{\rm b} = 1$.
In this case, $\gamma$ is given by \eqref{gammafunct} where we fix $p = 1000$ and $\Phi^{e}$ is taken in reference to $d$.  Additionally, define $\Psi$ as in \eqref{eq:3}.
Now an assured controller is given by Algorithm 1.

For the purpose of this study, we hypothesize an open-loop desired control input 
\begin{equation}\label{csdes}
    \cu^{\rm d}(t) = 
    \begin{bmatrix}
    -0.3\,\, \mbox{sin}(\pi t/4) \\
    \hspace{.285 cm} 0.2\, \mbox{cos}(\pi t/2)
    \end{bmatrix},
\end{equation}
and simulate the system \eqref{cssys} under the ASIF controller Algorithm 1, where we let $\alpha(\psi) = 1000 \psi^3$. Note that, even though the theory above was developed assuming a given desired closed-loop feedback controller, the same approach is applicable if an open-loop control input is provided instead as a function of time. A 4-second simulation is conducted using MATLAB 2020a and simulation results are provided in Figure \ref{fig_mohit}.  The system response is simulated via Euler integration with a time-step of $0.01$ seconds and the optimization problem Algorithm \ref{alg:one} is computed at each time-step using CVX, a convex optimization tool built for use with MATLAB.  In the case of this experiment, the average the solver time is $0.54$ seconds per optimization\footnote{The code for this experiment is publicly available on the GaTech FACTS Lab Github: https://github.com/gtfactslab/Abate\_CDC2020}.

In the simulation the assured controller $\cu^{\rm ASIF}$ drives the system \eqref{cssys} out of the safe set; however, the system remains in $S_{\rm b}^+(1)$ and all points along the system trajectory are safe with respect to $\cu^{\rm b}$.


\section{Conclusion}
This work presents a problem formulation for runtime assurance for control systems with disturbances, and a specific solution to the problem statement is presented, whereby the nondeterminism in the system model is accommodated via the mixed-monotonicity property.
The proposed assured controller computes an optimization problem containing only a finite number of affine constraints, and we demonstrate the applicability of our construction through a case study. 

\begin{figure}[t]
    \begin{subfigure}{0.49\textwidth}
        \input{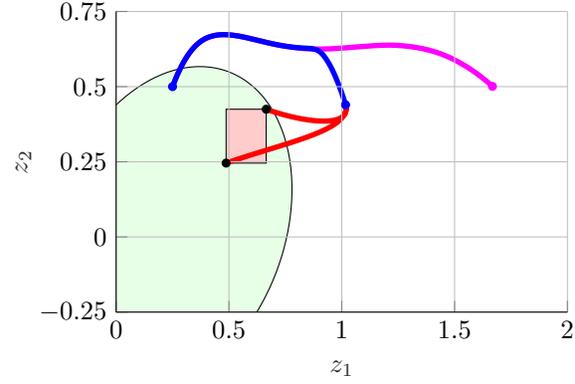}
        \caption{
        Cart displacement trajectory on time interval $[0,\, 4]$.  The nominal trajectory $\Phi(\,\cdot\,; (x_0,\, z_0),\, \cu^{\rm d},\, \mathbf{w})$ is shown in pink.
        The assured trajectory $\Phi^{\rm ASIF}(\,\cdot\,; (x_0,\, z_0),\, \mathbf{w})$ is shown in blue. 
        Bounds on the safe backup trajectory are computed via the decomposition function $d$, and are shown in red. $S^{\rm b}$ is shown in green at time $T = 4$.
        \vspace{.6cm}
        }
        \label{fig:ex1_1}
    \end{subfigure}
    ~
    \begin{subfigure}{0.49\textwidth}
        \input{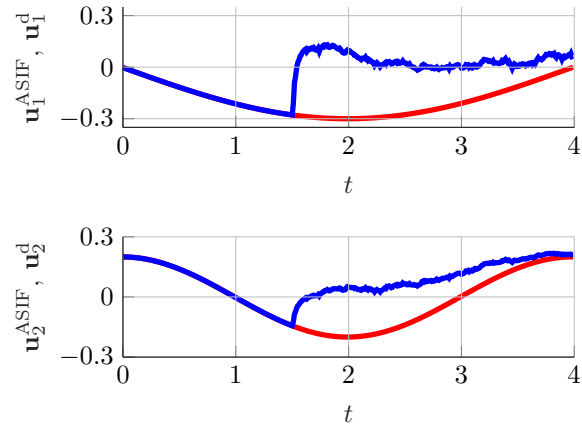}
        \caption{Control input signals vs. time.  The desired control input $\cu^{\rm d}$ from \eqref{csdes} is shown in red.  The applied input, which is chosen via Algorithm 1, is shown in blue.
        }
        \label{fig:ex1_3}
    \end{subfigure}
    \caption{  
    Implementing Algorithm 1 to assure the vehicle platoon \eqref{cssys}.  The carts begin with an initial velocity state $x_0 = [-1/4,\, 0,\, 1/2]^T$ and an initial displacement state $z_0 = [1/4,\, 1/2]^T$.  A random disturbance $\mathbf{w} : [0,\, 4] \rightarrow \W$ is also chosen.
    }
    \label{fig_mohit}
\end{figure} 
 
\bibliography{Bibliography}
\bibliographystyle{ieeetr}
\end{document}